%% file: paper.tex
\documentclass[a4paper, USenglish, reviewversion]{oasics-v2018}

\usepackage{auxlib}

\usepackage[
    giveninits=true,
    isbn=false,
    maxbibnames=99,
    sortcites=true,
    style=numeric,
    url=false,
]{biblatex}

\AtEveryBibitem{
    \clearfield{month}
    \clearlist{location}
    \clearlist{publisher}
    \clearname{editor}
}

\newcommand\bell{\ensuremath{\bm{\ell}}}
\newcommand\Bell{\ensuremath{\bm{L}}}

\def\li{\ell_i}
\def\lj{\ell_j}
\def\lo{\varnothing}
\def\lt{\round{\lo}}
\def\r{r_\ell(i,j)}

\def\h#1#2{\operatorname{h}_{#1}\left({#2}\right)}
\def\hh#1#2{\hat{\operatorname{h}}_{#1}\left({#2}\right)}

\let\BigO\LDAUOmicron         
\let\LittleO\LDAUomicron      
\let\BigOmega\LDAUOmega       

\let\Probability\Prob    
\let\Expected\Ex         

\newcommand{\Exp}[1]{\ensuremath{\exp\left(#1\right)}}
\def\hhb#1{\hh{b}{#1}}

\title{Simple Load Balancing}

\author{Petra Berenbrink}{Universität Hamburg, Germany}{petra.berenbrink@uni-hamburg.de}{}{}
\author{Tom Friedetzky}{Durham University, U.K.}{tom.friedetzky@dur.ac.uk}{https://orcid.org/0000-0002-1299-5514}{}
\author{Dominik Kaaser}{Universität Hamburg, Germany}{dominik.kaaser@uni-hamburg.de}{https://orcid.org/0000-0002-2083-7145}{}
\author{Peter Kling}{Universität Hamburg, Germany}{peter.kling@uni-hamburg.de}{https://orcid.org/0000-0003-0000-8689}{}

\authorrunning{P. Berenbrink, T. Friedetzky, D. Kaaser, and P. Kling}

\Copyright{%
    Petra Berenbrink,
    Tom Friedetzky,
    Dominik Kaaser, and
    Peter Kling
}

\subjclass{%
    Mathematics of computing → Probability and statistics → Stochastic processes
}


\category{}
\relatedversion{}
\supplement{}
\funding{}
\acknowledgements{}

\EventEditors{John Q. Open and Joan R. Access}
\EventNoEds{2}
\EventLongTitle{42nd Conference on Very Important Topics (CVIT 2016)}
\EventShortTitle{CVIT 2016}
\EventAcronym{CVIT}
\EventYear{2016}
\EventDate{December 24--27, 2016}
\EventLocation{Little Whinging, United Kingdom}
\EventLogo{}
\SeriesVolume{42}
\ArticleNo{}
\nolinenumbers
\hideOASIcs

\begin{document}
\maketitle
\begin{abstract}
{\input{01-abstract}}
\end{abstract}

{\input{10-introduction}}
{\input{11-model}}
{\input{20-analysis}}
{\input{21-phase1}}
{\input{22-phase2}}
{\input{23-phase3}}

\printbibliography
\end{document}

%% file: 01-abstract.tex
We consider the following load balancing process for $m$ tokens distributed
arbitrarily among $n$ nodes connected by a complete graph: In each time step a
pair of nodes is selected uniformly at random. Let $\ell_1$ and $\ell_2$ be
their respective number of tokens. The two nodes exchange tokens such that they
have $\ceil{(\ell_1 + \ell_2)/2}$ and $\floor{(\ell_1 + \ell_2)/2}$ tokens,
respectively. We provide a simple analysis showing that this process reaches
almost perfect balance within $\BigO{n\log{n} + n \log{\Delta}}$ steps, where
$\Delta$ is the maximal initial load difference between any two nodes.

%% file: 10-introduction.tex
\section{Introduction}\label{sec:introduction}

We consider a load balancing problem for $m$ tokens on $n$ identical nodes
connected by a complete graph. Each node starts with some number of tokens and
the objective is to distribute the tokens as evenly as possible. A natural and
simple process to reach this goal is as follows: At each time step, a pair of
nodes $(u, v)$ is chosen uniformly at random and their loads (number of tokens)
are balanced as evenly as possible. We provide a simple and elementary proof
that this process takes, \whpshort/, $\BigO{n \log{n} + n \log{\Delta}}$ time
steps to reach almost perfect balance. Here, $\Delta$ is the maximal initial
load difference between any two nodes\footnote{%
    We assume $\Delta > 0$ (such that $\log{\Delta}$ is well defined);
    otherwise the system is already perfectly balanced.
} and almost perfect balance means that all nodes have a load in
$\set{\round{\lo}-1, \round{\lo}, \round{\lo}+1}$, where $\lo = m/n$ and
$\round{\lo}$ is $\lo$ rounded to the nearest integer.

\paragraph{Related Work}
There is a vast body of literature on load balancing, even when considering
only theoretical results. As it is beyond the scope of this article to provide
a complete survey, we focus on results for discrete load balancing on complete
graphs and processes with sequential (or at least independent) load balancing
actions. For an overview of results on general graphs, processes with multiple
correlated load balancing actions (like the so-called \emph{diffusion model}),
and other variants we refer the reader to~\cite{conf/focs/SauerwaldS12,
journals/algorithmica/BerenbrinkFKMNW18}.

We should first like to note that the result we prove may almost be considered
folklore and variants of it have been proved in different contexts, for example
in~\cite{conf/nca/MocquardAS16} (who use this to prove results in a specific
distributed computational model called \emph{population model})
or~\cite{conf/focs/SauerwaldS12} (who study load balancing on general graphs;
see below). Nevertheless, we believe that this load balancing setting is
important enough (variants of it appearing as building blocks in many
distributed algorithms) that there is merit in providing a dedicated,
intuitive, and elementary proof.

A related load balancing model is the \emph{matching model}, also known as the
\emph{dimension exchange model}. Here, each time step an arbitrary matching of
nodes is given and any two matched nodes balance their load. In our case, the
matching in each round consists of a single edge chosen independently and
uniformly at random. A rather general way to analyze this model on arbitrary
graphs was introduced by \textcite{conf/focs/RabaniSW98}. The authors studied
how far the discrete load balancing process diverges from its continuous
counterpart (where tokens can be split arbitrarily). This idea was later
extended and used by \textcite{conf/focs/SauerwaldS12} to prove the currently
best bounds for the matching model (and others). For the complete graph and
assuming that each round a random matching is used, their results imply a bound
of $\LDAUOmicron{\log m}$ rounds (which translates to $\LDAUOmicron{n \log m}$
time steps in our model, as they use matchings of size $\LDAUTheta{n}$) to
reduce the difference between maximum and minimum load to some (unspecified)
constant. The time bound holds with probability $1 - \exp\bigl(-{(\ln
n)}^{\LDAUTheta{1}}\bigr)$, which is slightly weaker than our probabilistic
guarantee.

Another related strain of literature considers discrete, sequential load
balancing, but with the restriction that only one token can move per time step.
\Textcite{conf/podc/Goldberg04} considered a simple local search process in
this scenario: Tokens are activated by an independent exponential clock of rate
$1$. Upon activation, a token samples a random node and moves there if that
node's load is smaller than the load at the token's current host node. It has
recently been proved~\cite{conf/ipps/BerenbrinkKLM17} that this process reaches
perfect balance in $\LDAUOmicron{\log n + \log(n) \cdot n^2/m}$ time (both in
expectation and \whp/), which is asymptotically tight.

%% file: 11-model.tex
\subsection{Model and Notation}\label{sec:model}

Assume $m$ indistinguishable tokens are distributed arbitrarily among $n$ nodes
of a complete graph. Define the \emph{load vector} $\Bell(t) = \left(\ell_1(t),
\dots, \ell_n(t) \right) \in \Z^n$ at time $t$, where $\ell_i(t)$ is the number
of tokens (load) assigned to node $i$ at time $t$. The \emph{discrepancy}
$\Delta\Bell(t)$ at time $t$ is the maximal load difference between any two
nodes. Let $\Delta = \Delta\Bell(0)$ be the \emph{initial discrepancy}. We
define $\lo = m/n$ as the \emph{average load} and use $\lt$ to denote the
average load rounded to the nearest integer.

Given the load vector $\Bell(t)$ at time $t$, our load balancing process
performs the following actions during time step $t$:
\begin{enumerate*}
\item

    Two nodes $u$ and $v$ are selected uniformly at random without replacement.

\item

    Their loads are updated according to $\ell_u(t+1)= \ceil*{\left(\ell_u(t) +
    \ell_v(t)\right)/2}$ and $\ell_v(t+1) = \floor*{\left(\ell_u(t) +
    \ell_v(t)\right)/2}$.

\end{enumerate*}

For the sake of the analysis we assume that tokens are ordered (arbitrarily) on
each node. Based on this order, we define the \emph{height} $h_b(t)$ of a token
$b$ at time $t$ as the number of tokens that precede $b$ in this order. The
\emph{normalized height} $\hhb{t} = \h{b}{t} - \lt$ enumerates the tokens
relative to the rounded average $\lt$. Furthermore, we initially assume that
balancing operations between two nodes operate in \emph{stack mode}, where the
topmost tokens of the node with higher load are moved to the node with lower
load (see \cref{fig:balancing-simple}). For the second part of our analysis
(\nameref{sec:phase2}) we assume that balancing operations operate in
\emph{skip mode}, where every second token is moved (see
\cref{fig:balancing-every-second}). Finally, in the third part of our analysis
(\nameref{sec:phase3}), we assume that the excess tokens are first shuffled
before the balancing operates in stack mode. Note that the mode does not
influence the balancing process but merely facilitates the analysis.

\begin{figure}
\captionsetup[subfigure]{justification=centering}
\hfill
\begin{subfigure}{0.22\linewidth}
\centering
\includegraphics[width=\linewidth]{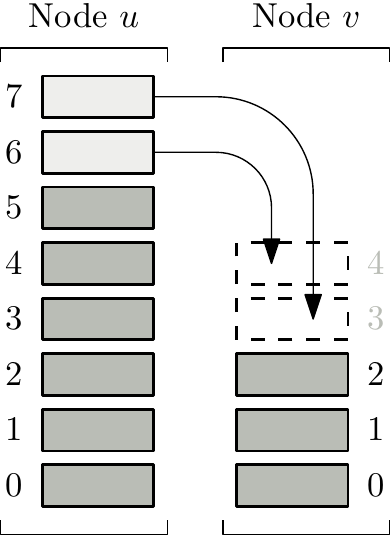}
\caption{Stack Mode}\label{fig:balancing-simple}
\end{subfigure}
\hfill
\begin{subfigure}{0.22\linewidth}
\centering
\includegraphics[width=\linewidth]{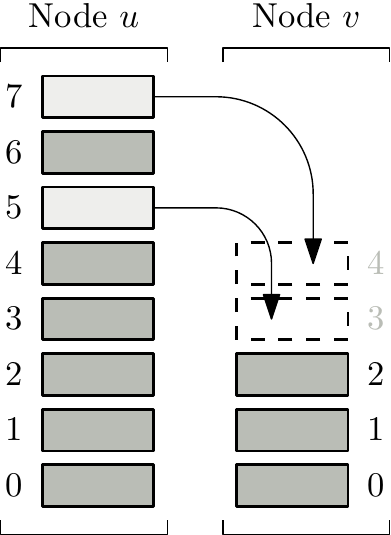}
\caption{Skip Mode}\label{fig:balancing-every-second}
\end{subfigure}
\hfill{}
\caption{%
    Illustration of the different modes assumed for balancing operation during
    the analysis.
}\label{fig:balancing}
\end{figure}

%% file: 20-analysis.tex
\section{Analysis}\label{sec:analysis}

We split the analysis into three phases. In \nameref{sec:phase1} we use a
potential function argument to show that, \whpshort/, it takes $\BigO{n\log{n}
+ n\log{\Delta}}$ time steps until at most $n/2$ nodes have a load larger than
$\lo + \LDAUTheta{1}$. In \nameref{sec:phase2} we look at individual tokens and
prove that, \whpshort/, it takes $\BigO{n\log{n}}$ more time steps until all
nodes have load at most $\lo + \LDAUTheta{1}$. Finally, in \nameref{sec:phase3}
we prove that, \whpshort/, it takes $\BigO{n\log{n}}$ further time steps until
the maximum load is at most $\lt + 1$. Using a symmetry-based argument we get a
similar bound on the minimum load and, thus, the following theorem.
\begin{theorem}\label{thm:main-result}
Let $\Bell(0) \in \mathbb{N}_{0}^{n}$ be the initial load vector of the load
balancing process on $n$ nodes and let $\Delta = \Delta\Bell(0)$ be the
initial discrepancy. Let furthermore $T$ be the first time when all nodes
have load in $\set{\lt-1, \lt, \lt+1}$. \Whp/, $T = \BigO{n \log \Delta + n
\log n}$.
\end{theorem}
Observe that \cref{thm:main-result} is tight for polynomial $\Delta$: with
constant probability there are nodes that are not selected at all during the
first $\LittleO{n \log{n}}$ time steps.

%% file: 21-phase1.tex
\subsection[Phase~1]{Phase 1: Potential Function Analysis}\label{sec:phase1}

We analyze the process with the potential function defined via
\begin{equation}
  \Phi(\bell)
= \sum_{i=1}^{n} {\left( \ell_i - \lo \right)}^2
\end{equation}
for a load vector $\bell \in \N_0^n$.
\begin{lemma}\label{lem:phase1}
Let $T_1$ be the first time step for which $\Phi(\Bell(T_1)) < n$. \Whpshort/, $T_1
= \BigO{n\log{n} + n \log{\Delta}}$.
\end{lemma}
\begin{proof}
We start by analyzing the expected change of the potential during one time
step. Let $\delta(\bell, i,j)$ be the potential drop of a fixed load vector
$\bell=(\ell_1, \dots,\ell_n)\in\N_0^n$ when nodes $i$ and $j$ are
balancing. Then
\begin{equation}
  \delta(\bell, i, j)
= {(\li - \lo)}^2 + {(\lj - \lo)}^2 - {\left(\ceil*{\frac{\li + \lj}{2}} - \lo\right)}^2 - {\left(\floor*{\frac{\li + \lj}{2}} - \lo\right)}^2
\enspace.
\end{equation}
We define the discretization error $\r$ as $1$ if $\li + \lj$ is odd and $0$
otherwise. This allows us to expand and simplify the above expression to get
\begin{equation}\label{eq:pot0}
\begin{aligned}
\delta(\bell, i, j)
&= \frac{{\left(\li-\lj\right)}^2 - \r^2}{2} \geq \frac{{\left(\li-\lj\right)}^2}{2} - 1/2
\enspace.
\end{aligned}
\end{equation}
\Cref{eq:pot0} implies that the potential never increases when two nodes
balance (the only negative term is $-\r^2/2$, but $\r = 1$ implies $\li \neq
\lj$ and, thus, ${(\li-\lj)}^2 \geq 1$). We now calculate the expected
potential after one time step. Each pair of nodes is chosen uniformly at random
with probability $1 / \binom{n}{2}$. When chosen, the potential drops by
$\delta(\bell(t),i,j)$. Therefore,
\begin{equation}
\begin{aligned}
      \Expected{\Phi(\Bell(t+1)) | \Bell(t) = \bell}
&=    \sum_{i = 1}^{n} \sum_{j = i+1}^{n} \frac{1}{\binom{n}{2}}\cdot\left(\Phi(\bell) - \delta(\bell,i,j)\right)\\
&\leq \Phi(\bell) - \frac{1}{2 \binom{n}{2}} \sum_{i = 1}^{n} \sum_{j = i+1}^{n}{\left( \li - \lj \right)}^2 + \frac{1}{2}
\enspace.
\end{aligned}
\end{equation}
We now use
$\sum_{i = 1}^{n} \sum_{j = i+1}^{n}{\left( \li - \lj \right)}^2
= n \cdot \Phi(\bell)$
and obtain
\begin{equation}\label{eq:expected-drop}
\begin{aligned}
\Expected{\Phi(\Bell(t+1)) | \Bell(t) = \bell}
&\leq \Phi(\bell) - \frac{1}{2 \binom{n}{2}} \cdot n\cdot\Phi(\bell) + \frac{1}{2}
\leq \left( 1 - \frac{1}{n} \right) \cdot \Phi(\bell) + \frac{1}{2}
\enspace.
\end{aligned}
\end{equation}

We now partition the time horizon into \emph{rounds} of $n$ consecutive time
steps each and look at \emph{successful} rounds (in which the potential drops
sufficiently). We then argue that $\BigO{\log\left(\Phi(\Bell(0))/n\right)}$
successful rounds suffice for the potential to drop below $n$ and that,
\whpshort/, we have this many successful rounds among the first $\BigO{\log{n}
+ \log{\Delta}}$ rounds.

Let round $r$ consist of the time steps in $\intco{r \cdot n, (r + 1) \cdot
n}$. We assume that the load vector $\Bell(r \cdot n) = \bell$ at the beginning
of round $r$ is fixed. By recursive application of \cref{eq:expected-drop}, we
get
\begin{equation}
\begin{aligned}
      \Expected{\Phi(\Bell((r+1) \cdot n)) | \Bell(r \cdot n) = \bell}
&\leq {\left( 1 - \frac{1}{n} \right)}^n \cdot \Phi(\bell) + \frac{1}{2}\cdot \sum_{i=0}^{n-1}{\left(1-\frac{1}{n}\right)}^i \\
&\leq e^{-1} \cdot \Phi(\bell) + \frac{n}{2}\cdot \left(1-e^{-1}\right)
\enspace,
\end{aligned}
\end{equation}
where we used the inequality ${\left(1-1/n\right)}^n \leq e^{-1}$.
As long as $\Phi(\bell) \geq n$, the last expression is at most
\begin{equation}
     \left(e^{-1} + \frac{n}{2 \Phi(\bell)} \left(1-e^{-1}\right)\right) \cdot \Phi(\bell)
\leq \frac{1 + e^{-1}}{2} \cdot \Phi(\bell)
<    \frac{3}{4} \Phi(\bell)
\enspace.
\end{equation}
Applying the Markov inequality now gives us, for an $\bell$ with $\Phi(\bell) \geq n$,
\begin{equation}\label{eq:drop-prob}
\begin{aligned}
\Probability{\Phi(\Bell((r+1)\cdot n)) \geq \frac{7}{8}\cdot \Phi(\bell) | \Bell(r \cdot n) = \bell}
&\leq \frac{6}{7}
\enspace.
\end{aligned}
\end{equation}

We define a round $r$ to be \emph{successful} if $\Phi(\Bell((r+1)\cdot n))
\leq 7/8 \cdot \Phi(\Bell(r \cdot n)) ~ \lor ~ \Phi(\Bell(r \cdot n)) < n$ and
use $\cE_r$ to denote this event. \Cref{eq:drop-prob} implies
$\Probability{\cE_r} \geq 1/7$.

We now argue that after at most $\rho = \log_{8/7} \left( \Phi(\Bell(0))/n
\right)+1$ successful rounds the potential is smaller than $n$. Let
$r_\rho$ be the $\rho$-th successful round. There are two cases. If there
exists a round $r \leq r_\rho$ for which $\Phi(\Bell(r \cdot n)) < n$, then
$\Phi(\Bell(r_\rho)) < n$ is trivially true since the potential does never
increase. Otherwise, by definition of a successful round, after
$\rho$ successful rounds we have
\begin{equation}
     \Phi(\Bell(r_\rho \cdot n))
\leq {\left(\frac{7}{8} \right)}^\rho \cdot \Phi(\Bell(0))
=    \frac{7}{8} \cdot \frac{n}{\Phi(\Bell(0))} \cdot \Phi(\Bell(0))
<    n
\enspace.
\end{equation}

It remains to show that, \whpshort/, during the first $\BigO{\log{n} +
\log{\Delta}}$ rounds at least $\rho$ rounds are successful. Let the random
variable $X$ denote the number of successful rounds during the first
$168\left(\ln{n} + \log\Delta\right)$ rounds. Since each round is successful
with probability at least\footnote{%
    While the rounds are not independent, the lower bound holds independently
    for each round.
} $1/7$, the random variable $X$ stochastically dominates the binomial random
variable $Y \sim \BinDistr(168\left(\ln{n} + \log\Delta\right), 1/7)$ (written
$X \succeq Y$). Applying Chernoff bounds to $Y$ with its expected value
$\mu = \Expected{Y} = 24\left(\ln{n} + \log\Delta\right)$ gives
\begin{equation}
\begin{aligned}
     \Probability{Y \leq \rho}
&=   \Probability{Y \leq \left(1-\frac{\mu - \rho}{\mu}\right)\mu}
\leq \Exp{-\frac{{\left(\mu - \rho \right)}^2}{\mu^2} \cdot \frac{\mu}{2}}
\stackrel{\eqref{eq:mu-rho}}{\leq}
     \Exp{- \frac{\mu}{8}}
\leq n^{-3}
\enspace,
\label{eq:T1-whp}
\end{aligned}
\end{equation}
where we used the following inequality to bound $\mu - \rho$, holding for
$\Delta \geq 1$:
\begin{equation}
\rho
=    \log_{8/7}{\left(\frac{\Phi(\Bell(0))}{n}\right)}+1
\leq \log_{8/7}{\left(\frac{n\cdot\Delta^2}{n}\right)}+1
\leq \frac{2 \log{\Delta}}{ \log{8}/{7} } + 1
<    12 \log\Delta+1
<    \frac{\mu}{2}
\enspace.
\label{eq:mu-rho}
\end{equation}
Since $X \succeq Y$, \cref{eq:T1-whp} implies that the probability of having
fewer than $\rho$ successful rounds during the first $7n\cdot\mu$ time steps is
smaller than $n^{-3}$. Therefore, \whpshort/, $T_1 \leq 7n \cdot \mu =
\BigO{n\log{n} + n\log{\Delta}}$.
\end{proof}

%% file: 22-phase2.tex
\subsection[Phase~2]{Phase 2: Improving Individual Tokens}\label{sec:phase2}

We now consider individual tokens. We start our analysis with
\cref{lem:items-improve}, where we show that during any time step any token
with normalized height larger than some constant reduces its height with
probability $\BigOmega{1/n}$ by a constant factor. This is then used in
\cref{lem:phase2} to argue that it takes at most $\BigO{n \log{n}}$ time steps
for all tokens to reach a constant normalized height.

For the sake of the analysis we now define which tokens are selected to be
transferred when two nodes are balanced. Recall that according to the
definition of the process tokens are indistinguishable and therefore arbitrary
tokens may be selected.

Fix a time step $t$ and assume that node $u$ interacts with node $v$. In order
to balance their loads, we need to move tokens from the node with larger load
to the node with smaller load (say from $u$ to $v$). To do so, we start with
the token at maximal height and take every other token until we have selected
required number of tokens. Then we place all tokens on node $v$ in their
original order. An example for this process is sketched in
\cref{fig:balancing-every-second}.

For the remainder, let $c \geq 10$ be a constant and recall that $T_1$ is the
first time step of the second phase. The rule defined above allows us to show
the following lemma.
\begin{lemma}\label{lem:items-improve}
Let $t \geq T_1$ and let $b$ be a token with normalized height $\hhb{t} > 2c$.
Then $\hhb{t+1} \leq 17/20 \cdot \hhb{t}$ with probability at least~$1/n$.
\end{lemma}
\begin{proof}
The idea of the proof is as follows. We first argue that at any time after the
first phase fewer than half of the nodes have load larger than or equal to
$\lo +c$. This is then used to derive a lower bound on the probability that a
token of normalized height larger than $2c$ takes part in balancing with a node
that has load at most $\lo + c$. Finally, we compute the new height of the
token, which yields the lemma.

We now give the formal proof. Let $S(t) = \set{ v | \ell_v(t) \geq \lo + c}$ be
the set of nodes which have load at least $\lo + c$ and suppose that
$\abs{S(t)} \geq n/2$. Then
\begin{equation}
     \Phi(\Bell(t))
=    \sum_{i=1}^{n} {\left(\ell_i(t) - \lo\right)}^2
\geq \sum_{i \in S} {\left(\ell_i(t) - \lo\right)}^2
\geq \sum_{i \in S} c^2
\geq 100 n/2
>    n
\enspace.
\end{equation}
However, the potential function does not increase over time and, thus,
\cref{lem:phase1} implies that $\Phi(\Bell(t)) \leq n$ for any $t \geq T_1$.
This is a contradiction and, therefore, $\abs{S(t)} < n/2$.

We now proceed to lower bound the probability that $b$ reduces its normalized
height by a constant factor. Let $i$ be the node on which token $b$ is stored
at time $t$. With probability $2/n$, node $i$ is selected as one of the two
nodes for balancing. Let furthermore $j$ be the other node selected for
balancing. Since $\abs{S(t)} < n/2$, node $j$ has load at most $\lo + c$ with
probability at least $1/2$ (independent of $i$'s selection). In that case,
either $\floor*{\frac{\ell_i(t) - \ell_j(t)}{2}}$ or $\ceil*{\frac{\ell_i(t) -
\ell_j(t)}{2}}$ tokens are moved, depending on whether $(i,j)$ or $(j,i)$ are
selected. Using that each other token is moved (see
\cref{fig:balancing-every-second}), carefully bounding the new height gives in
both cases, regardless of whether $b$ is transfered to node $j$ or stays on
node $i$, that the new height of token $b$ becomes at most
\begin{equation}
     h_b(t+1)
\leq \lj(t) + \ceil*{\frac{h_b(t) - \lj(t) + 1}{2}} + 1
\leq \lj(t) + \frac{h_b(t) - \lj(t)}{2} + 2
=    \frac{h_b(t) + \lj(t) + 4}{2}
\enspace.
\end{equation}
We now bound the ratio between the new and the old normalized height of token
$b$. For $\lj(t) \leq \lt + c$ and $h_b(t) \geq \lt + 2c$, this ratio is at
most
\begin{equation}
     \frac{\hhb{t+1}}{\hhb{t}}
=    \frac{\frac{1}{2}\left(h_b(t) + \lj(t) + 4\right) - \lt}{h_b(t) - \lt}
=    \frac{1}{2} + \frac{1}{2}\cdot \frac{\lj(t) - \lt + 4}{h_b(t) - \lt}
\leq \frac{1}{2} + \frac{c + 4}{4c}
\leq 0.85
\enspace ,
\end{equation}
where the last inequality holds since $c \geq 10$.
Therefore, at any time $t \geq T_1$ and for any token $b$ with
$\hhb{t} \geq 2c$, we have $\hhb{t+1} \leq 0.85 \cdot \hhb{t}$ with
probability at least $1/n$.
\end{proof}

We are now ready to show the main lemma for the second phase.
\begin{lemma}\label{lem:phase2}
Let $T_2$ be the first time for which
$\displaystyle \max_{1\leq i \leq n}\set{\ell_i(T_2)} \leq \lo + 2c$ and
$\displaystyle \min_{1\leq i \leq n}\set{\ell_i(T_2)} \geq \lo - 2c$.
\Whp/, $T_2 = T_1 + \BigO{n \log{n}}$.
\end{lemma}
\begin{proof}
We first show the claim for the maximal load and then use a coupling argument
to extend the analysis to the minimal load. For the maximal load, we consider a
fixed token $b$ and use \cref{lem:items-improve} to define and bound the
probability of a \emph{successful} time step w.r.t.\ $b$. Then we show that
this event occurs sufficiently often during the first $\BigO{n \log{n}}$ time
steps such that $b$ reaches normalized height at most $2c$ \whp/. Finally, we
show the claim by a union bound over all tokens of normalized height larger
than $2c$.

Let $b$ be an arbitrary but fixed token with $\hhb{t} \geq 2c$. We call a time
step $t$ \emph{successful} if $\hhb{t+1} \leq 17/20\cdot \hhb{t} ~ \lor ~
\hhb{t} \leq 2c$. From \cref{lem:items-improve} we get that time step $t$ is
successful with probability at least $1/n$. Note that while the behavior of two
different tokens may be highly correlated, for one fixed token the lower bounds
hold independently for any time step in the second phase. This allows us to
leverage stochastic dominance of a binomial distribution as follows: Let the
random variable $X_b(\tau)$ denote the number of successful time steps during
the first $\tau$ time steps in the second phase. Since each time step is
successful with probability at least $1/n$, the random variable $X_b(\tau)$
stochastically dominates the binomial random variable $Y_b(\tau) \sim
\BinDistr(\tau, 1/n)$. Applying Chernoff bounds to $Y_b(\tau)$ with $\tau =
12n\log{n}$ gives
\begin{equation}
     \Probability{Y_b(12n\log{n})
\leq \left(1-\frac{3}{4}\right)\Expected{Y_b(12n\log{n})}}
\leq \Exp{-\frac{1}{2}\cdot\frac{9}{16}\cdot 12\log{n}}
\leq n^{-3}
\enspace.
\end{equation}
With the above mentioned stochastic dominance $X_b(\tau) \succeq Y_b(\tau)$, we
get that $X_b(12 n \log{n}) \leq 3 \log n$ with probability at most $ n^{-3}$.
It remains to show that the normalized height of $b$ after $3\log{n}$
successful time steps is at most
$2c$. Observe that $\hhb{T_1} \leq \sqrt{n}$, since otherwise $\Phi(\Bell(T_1))
\geq n$. Therefore, after at most $3\log{n}$ successful time steps in the
second phase, the normalized height of $b$ is at most\footnote{%
    The maximum covers the fact that the analysis does not extend to $\hhb{t} <
    2c$.
}
\begin{equation}
     \hhb{T_1 + 12n\log{n}}
\leq \max\set{\sqrt{n} \cdot {\left(\frac{17}{20}\right)}^{3\log n}, 2c}
\leq 2c
\enspace.
\end{equation}

We now use the union bound on the above analysis over all tokens as follows.
From the bound on the potential function in \cref{lem:phase1} we obtain that
after the first phase at most $n$ tokens remain above the average, since
otherwise the potential would be larger than $n$. Observing that the height of
a token never increases and taking the union bound over all tokens of
normalized height above $2c$ gives us that all tokens have remaining height at
most $2c$ after at most $12 n\log{n}$ interactions with probability
$1-1/n^{-2}$.

We now argue an analogous bound for the minimal load. Let $\bell \in
\mathbb{Z}^{n}$ be the initial load vector of the load balancing process
$\Bell(0) = \bell , \Bell(1), \Bell(2), \dots$ and let $-\bell$ be the initial
load vector of the load balancing process $\Bell'(0) = -\bell, \Bell'(1),
\Bell'(2), \dots$. We can couple the processes such that whenever a pair of
nodes $(u, v)$ is chosen in $\Bell(t)$, the pair of nodes $(v, u)$ is chosen in
$\Bell'(t)$. This coupling ensures (determinstically) that $\ell_i(t) = -
\ell'_i(t)$
and, thus, implies $\Probability{\ell_i(t) = x} = \Probability{\ell'_i(t) =
-x}$. By applying the upper bound on the maximal load to $\Bell'(T_1 +
12n\log{n})$, we get a lower bound on the minimal load in $\Bell(T_1 +
12n\log{n})$. Thus, $T_2 \leq T_1 + \BigO{n\log{n}}$, which concludes the
proof.
\end{proof}

%% file: 23-phase3.tex
\subsection[Phase~3]{Phase 3: Fine Tuning}\label{sec:phase3}

For the sake of the analysis of the third phase, we use the following rule to
select tokens to transfer when balancing two nodes. We again assume that nodes
operate like stacks, with the following additional rule: both nodes shuffle
their tokens of normalized height in $\set{2, 3, \dots, 2c}$ (if they exist)
before balancing the loads. This rule allows us to show the following lemma,
our main result.

\begin{lemma}\label{lem:phase3}
Let $T_3$ be the first time for which
\begin{math}
\max_{1\leq i \leq n}\set{\ell_i(T_3)} \leq \lt + 1
\end{math}
and for which
\begin{math}
\min_{1\leq i \leq n}\set{\ell_i(T_3)} \geq \lt - 1
\end{math}.
\Whp/, $T_3 = T_2 + \BigO{n \log{n}}$.
\end{lemma}
\begin{proof}
We again start by analyzing the maximal load. We first show that at any time
step after the second phase at least a constant fraction of nodes has load at
most $\lt$. Then we consider an arbitrary but fixed token $b$ with $\hhb{t} >
1$ at time $t$ and show that with probability $\BigOmega{1/n}$ we have
$\hhb{t+1} \leq 1$. This is used to show that, \whpshort/, $\hhb{\tau} \leq 1$
for $\tau = \BigO{n \log n}$. The claim then follows from a union bound over
all tokens above normalized height $1$.

Fix a time step $t \geq T_2$ and let $\gamma$ be the fraction of nodes that
have load at most $\lt$ at time~$t$. We use the definition of the rounded
average load and \cref{lem:phase2} to compute
\begin{equation}
\begin{aligned}
&
n \cdot \left( \lt + 0.5 \right)
\geq n \cdot \lo
=    \sum_{1 \leq i \leq n} \li(t)
=    \smashoperator[r]{\sum_{ \li(t) > \lt}} \li(t) + \smashoperator[lr]{\sum_{\li(t) \leq \lt}} \li(t)
\\{}\geq{}&
     \smashoperator[r]{\sum_{\li(t) > \lt}} \left(\lt + 1\right) + \smashoperator[lr]{\sum_{\li(t) \leq \lt}} \left(\lt - 2c\right)
\geq n\cdot\left(1 - \gamma\right)\cdot \left(\lt + 1\right) + n\cdot\gamma\cdot\left(\lt - 2c\right)
\enspace.
\end{aligned}
\end{equation}
Therefore, $\gamma \geq {1}/({4c+2})$ is a constant.

Similar to the analysis of the second phase, we now consider an arbitrary but
fixed token $b$. Fix a time step $t \geq T_2$ and a token $b$ with
$\hhb{t} > 1$. Let $i$ be the node on which $b$ resides before time step $t$.
We have the following events.
\begin{enumerate}[nolistsep]
\item

    Node $i$ is selected for balancing: in any time step, $i$ is selected with
    probability $2/n$.

\item

    Token $b$ becomes the top-most token: all tokens $b'$ on node $i$ of
    normalized height $\hh{b'}{t} > 1$ are shuffled. Since there exist at most
    $2c$ such tokens after the second phase, $b$ becomes the top-most token
    with probability at least $1/(2c)$.

\item

    The other node has load at most $\lt$: since the fraction of such nodes is
    least $\gamma$, such a node is selected as the balancing partner with
    probability at least $\gamma$.

\end{enumerate}
We say $b$ is \emph{successful} in time step $t$ if all three of these events
occur. Observe that in this case $\hhb{t+1} \leq 1$. Let $p_{b}(t)$ be the
probability of a successful time step. Combining above probabilities, we get
$p_{b}(t) \geq 2/n \cdot 1/(2c) \cdot \gamma = \BigOmega{1/n}$.

We now consider $\BigO{n \log n}$ time steps after the second phase. Token $b$
is not successful at least once during these time steps with probability
\begin{equation}
     \prod_{t=1}^{\BigO{n\log{n}}}\left(1 - p_{b}(t)\right)
\leq {\left(1-\BigOmega{\frac{1}{n}}\right)}^{\BigO{n\log{n}}}
\leq n^{-\BigOmega{1}}
\enspace.
\end{equation}
That is, for a suitable choice of constants, $b$ reaches height
$1$ after at most \BigO{n\log{n}} time steps with probability $1 - 1/n^3$.
The upper bound on the load now follows from a union bound, since at most
$2c\cdot n$ tokens have normalized height above $1$ after the second phase. For
the lower bound on the load, precisely the same argument as in the proof of
\cref{lem:phase2} can be used.
\end{proof}

The proof of \cref{thm:main-result} now follows from a union bound over the
results from \cref{lem:phase1} for the first phase, \cref{lem:phase2} for the
second phase, and \cref{lem:phase3} for the third phase.